%% file: sum-set-simulation.tex
\documentclass[11pt]{article}

\usepackage[letterpaper,margin=1in]{geometry}
\usepackage{sn-preamble}

\newcommand{\rnote}[1]{\footnote{\color{purple}\textbf{Rocco:} {#1}}}

\newcommand{\shivam}[1]{{\color{blue} \textbf{Shivam:} {#1}}}

\newcommand{\ignore}[1]{}

\newcommand{\maj}{\mathop{\mathrm{maj}\/}}

\usepackage[boxed]{algorithm2e}
\usepackage{algpseudocode}

\SetAlCapSkip{0.5em}

\newcommand{\GL}{\textsc{Goldreich--Levin}}

\newcommand{\ConstructDT}{\textsc{ConstructDT}}
\newcommand{\ConstructImplicitDT}{\textsc{ConstructImplicitDT}}

\newcommand{\Sumset}{\textsc{Simulate-Sumset}}
\newcommand{\ImplicitSumset}{\textsc{Implicit-Simulate-Sumset}}

\newcommand{\ImplicitGL}{\textsc{Implicit-GL}}
\newcommand{\BLR}{\textsc{Linearity-Test}}

\newcommand{\Tsum}{\calT_\mathrm{sum}}

\title{Approximating Sumset Size
}

\author{Anindya De \vspace{3pt}\\
\small{\sl University of Pennsylvania} \and Shivam
Nadimpalli \vspace{3pt} \\ \small{\sl Columbia University} \and Rocco A.
Servedio \vspace{3pt}  \\ \small{\sl Columbia University}\vspace*{15pt}
}

\begin{document}

\pagenumbering{gobble}

\maketitle

\input{sections/abstract}

\newpage 
\setcounter{page}{1}
\pagenumbering{arabic} 

\input{sections/intro}

\input{sections/prelims}
\input{sections/not-so-glorious-algo}

\input{sections/tbil-km-rocco-version}

\input{sections/conclusion}

\section*{Acknowledgements}

A.D.~is supported by NSF grants CCF-1910534, CCF-1926872, and CCF-2045128.  S.N.~is supported
by NSF grants CCF-1563155 and by 
CCF-1763970.  R.A.S.~is supported by NSF grants CCF-1814873, IIS-1838154,
CCF-1563155, and by the Simons Collaboration on Algorithms and Geometry.  This
material is based upon work supported by the National Science Foundation under
grant numbers listed above. Any opinions, findings and conclusions or
recommendations expressed in this material are those of the authors and do not
necessarily reflect the views of the National Science Foundation (NSF). 

\bibliography{allrefs}{}
\bibliographystyle{alpha}
	
\end{document}

%% file: sections/abstract.tex

\begin{abstract}

Given a subset $A$ of the $n$-dimensional Boolean hypercube $\F_2^n$, the \emph{sumset} $A+A$ is the set $\{a+a': a, a' \in A\}$ where addition is in $\F_2^n$.  Sumsets play an important role in additive combinatorics, where they feature in many central results of the field.

The main result of this paper is a sublinear-time algorithm for the problem of \emph{sumset size estimation}. In more detail, our algorithm is given oracle access to (the indicator function of) an arbitrary $A \subseteq \F_2^n$ and an accuracy parameter $\eps > 0$, and with high probability it outputs a value $0 \leq v \leq 1$ that is $\pm \eps$-close to $\Vol(A' + A')$ for some perturbation $A' \subseteq A$ of $A$ satisfying $\Vol(A \setminus A') \leq \eps.$  It is easy to see that without the relaxation of dealing with $A'$ rather than $A$, any algorithm for estimating $\Vol(A+A)$ to any nontrivial accuracy must make $2^{\Omega(n)}$ queries. In contrast, we give an algorithm whose query complexity depends only on $\eps$ and is completely independent of the ambient dimension $n$.

\end{abstract}

%% file: sections/intro.tex

\section{Introduction}

Recent decades have witnessed a paradigm shift in the notion of what constitutes an ``efficient algorithm'' in algorithms and complexity theory.  
Motivated both by practical applications and theoretical considerations, the traditional gold standard of linear time as the ultimate benchmark for algorithmic efficiency has given way to the notion of \emph{sublinear-time} and \emph{sublinear-query} algorithms, as introduced by Blum and Kannan~\cite{BlumKannan:89} and Blum, Luby and Rubinfeld~\cite{BLR93}.  The study of sublinear algorithms is flourishing, with deep connections to many other areas including PCPs, hardness of approximation, and streaming algorithms (see e.g.~the surveys~\cite{Rubinfeld:06survey, Goldreich17book, Fischer, R00}).

The current paper is at the confluence of two different lines of research in the area of sublinear algorithms:

\begin{enumerate}

\item The first strand of work deals with sublinear algorithms to \emph{approximately compute (numerical-valued) functions on various combinatorial objects}. Example problems of this sort include (i) estimating the weight of a minimum spanning tree~\cite{ChazelleRT05}; (ii) approximating the minimum vertex cover size in a graph~\cite{PARNAS2007183}; and (iii) approximating the number of $k$-cliques in an undirected graph~\cite{EdenRonSeshadri}.  We note that for the first two of these results, the number of local queries that are made to the input combinatorial object is completely independent of its size.

\item The second strand of work is on {\em property testing of Boolean-valued functions}. Given a class of Boolean-valued functions ${\cal C}$, a testing algorithm for ${\cal C}$ is a query-efficient procedure which, given oracle access to an arbitrary Boolean-valued function $f$, distinguishes between the two cases that (i) $f$ belongs to class $\mathcal{C}$, versus (ii) $f$ is $\epsilon$-far from every function in $\mathcal{C}$. Flagship results in this area include algorithms for linearity testing~\cite{BLR93}, testing of low-degree polynomials ~\cite{RS96, jutpatrudzuc04}, junta testing~\cite{FKRSS03, Blaisstoc09}, and monotonicity testing~\cite{GGLRS, KhotMinzer15}.  Here too, for the first three of these properties, the query complexity of the testing algorithms depend only on the accuracy parameter $\eps$ and are completely independent of the ambient dimension $n$ of the function $f$.
\end{enumerate}

In recent years, a nascent line of work has emerged at the intersection of these two strands, where the high-level goal is to approximately compute various numerical parameters of Boolean-valued functions.  As an example, building on the work of Kothari {et~al.}~\cite{KothariNOW14}, Neeman~\cite{Neeman2013surface} gave an algorithm to approximate the ``surface area'' of a Boolean-valued function on $\R^n$, which is a fundamental measure of its complexity~\cite{KOS:08}. The~\cite{Neeman2013surface} algorithm has a query complexity of $\mathrm{poly}(S)$ if the target surface area is $S$, which is completely independent of the ambient dimension $n$.  Fitting the same motif is the work of Ron {et~al.}~\cite{RonWinfluence} who studied the problem of approximating the ``total influence'' (or equivalently, ``average sensitivity'') of a Boolean function. They showed that the optimal query complexity to approximate the influence $\Inf[f]$ of an arbitrary $n$-variable Boolean function $f$ to constant relative error is $\Theta(n/\Inf[f])$, and that this can be strengthened to essentially $\sqrt{n}/\Inf[f]$ for monotone functions. More recently, in closely related work Rubinfeld and Vasiliyan~\cite{DBLP:conf/approx/RubinfeldV19} have given a constant-query algorithm to approximate the ``noise  sensitivity'' of a Boolean function.

We note that each of the above three numerical parameters --- surface area, total influence, and noise sensitivity --- is essentially a measure of the ``smoothness'' of the Boolean function in question.  In contrast, in this work we are interested in the \emph{sumset size}, which has a rather different flavor and, as discussed below, is intimately connected to the subspace structure of the function.

\paragraph{Sumsets.} Let $A \subseteq \mathbb{F}_2^n$ be an arbitrary subset (which may of course be viewed as a Boolean function by considering its $\zo$-valued characteristic function). One of the most fundamental operations on such a set $A$ is to consider the {\em sumset} $A+  A$, defined as 
\[
A + A := \{x + y: x, y  \in A \}. 
\]
Here `$+$' is the group operation in $\mathbb{F}_2^n$. 
Note that for $A$ an affine subspace we have that $|A + A| = |A|$, and the converse (the only sets $A$ for which $|A +A|=|A|$ are affine subspaces) is also easily seen to hold.  In fact,  something significantly stronger is true: 
The celebrated Freiman--Ruzsa theorem~\cite{freiman1973foundations,AST_1999__258__323_0,apde.2012.5.627} states that if $|A+A| \le K \cdot|A|$, then $A$ is contained inside an affine subspace $H$ such that $|H| \le O_K(1) \cdot |A|$. Thus, the value of $|A+A|$ \emph{vis-a-vis} $|A|$ can be seen as a measure {the ``subspace structure''} of $A$.

\subsection{The Question We Consider} 

For $A\sse\F_2^n$, we define $\Vol(A) \coloneqq |A|/2^n \in [0,1]$ to be the \emph{normalized size} or \emph{volume} of $A$.  This paper is motivated by the following basic algorithmic problem about sumsets:

\begin{quote}
{\bf Sumset size estimation (naive formulation):}
Given black-box oracle access to a set $A \subseteq \F_2^n$ (via its characteristic function 
$A: \mathbb{F}_2^n \rightarrow \{0,1\}$), can we estimate the $\Vol(A + A)$ while making only ``few'' oracle calls to $A$?
\end{quote}

At first glance this seems to be a difficult problem, since to confirm that a given point $z$ does not belong to $A + A$ we must verify that at least one of $x,y \notin A$ for each of the $2^n$ pairs $(x,y)$ satisfying $x+y =z$. Indeed, for the above naive problem formulation, any algorithm must make $2^{\Omega(n)}$ queries even to distinguish between the two extreme cases that $\Vol(A+A)=0$ (i.e.~$A=\emptyset$) versus $\Vol(A+A) = 1 - \exp(-\Theta(n))$. To see this, suppose that $A$ is a uniform random subset of $2^{0.51n}$ many elements from $\F_2^n$. It is clear that any algorithm will need $\Omega(2^{0.49n})$ queries to distinguish such an $A$ from the empty set, and an easy calculation shows that such a random $A$ will with extremely high probability have $\Vol(A+A) = 1 - \exp(-\Theta(n)).$

This simple example already shows that some care must be taken to formulate the ``right'' version of the sumset size estimation problem.  This situation is analogous to the surface area testing problem that was studied in~\cite{KothariNOW14,Neeman2013surface}: In that setting, given oracle access to any set $A$, by adding a measure zero set $R$ to $A$ (which is undetectable by an algorithm with oracle access to $A$) it is possible to ``blow up'' the surface area of $A \cup R$ to an arbitrarily large value.  Thus the goal in~\cite{KothariNOW14, Neeman2013surface} is to find a value $S$ such that $\mathrm{surf}(A) \le S \le \mathrm{surf}(B)$ for a set $B$ that is ``close to $A$.''  Note that for surface area, it may be possible to dramatically increase the surface area of a set $A$ either by adding a small subset of new points or removing a small subset of existing points from $A$. In contrast, for sumset size it is clear that removing points from $A$ can never cause the sumset size to increase, and moreover  adding a small (random) collection $R \subseteq \F_2^n$ of $2^{0.51n}$ points to $A$ can always cause $\Vol((A \cup R) + (A \cup R))$ to become extremely close to 1.  Hence for our sumset size estimation problem we only allow \emph{subsets} of $A$ as the permissible ``close to $A$'' sets.

We thus arrive at the following formulation of our problem:

\begin{quote}
{\bf Sumset size estimation:}
Given black-box oracle access to a set $A \subseteq \F_2^n$ and an accuracy parameter $\eps >0$, compute $\Vol(A' + A')$ to additive accuracy $\pm \eps$ for some subset $A' \subseteq A$ which has $\Vol(A\setminus A') \leq \eps.$
\end{quote}

\subsection{Motivation} \label{sec:motivation}

Given the importance of sumsets in additive combinatorics, we feel that it is natural to investigate algorithmic questions dealing with basic properties of sumsets; estimating the size of a sumset is a natural algorithmic question of this sort.  
We further remark that while there is no direct technical connection to the present work, the path which led us to the sumset size estimation problem originated in an effort to develop a query-efficient algorithm for \emph{convexity testing} (i.e.~testing whether a subset $S \subseteq \R^n$ is convex versus far from convex, where the standard Normal distribution ${\cal N}(0,1)^n$ provides the underlying distance measure on $\R^n$).  In particular, the recent characterization by Shenfeld and van Handel of equality cases for the Ehrhard--Borell inequality (see Theorem~1.2 of \cite{rvh-equality}) implies that a closed symmetric set $S \subseteq \R^n$ is convex if and only if the Gaussian volume of $S$ equals the Gaussian volume of ${\frac {S + S} 2}$.   We believe that a robust version of this theorem might be useful for convexity testing; this naturally motivates a Gaussian space version of the sumset size estimation question, where now the Minskowski sum of sets in $\R^n$ plays the role of sumsets over $\F_2^n$.  We hope that the ideas and ingredients in the current work may eventually be of use for the Gaussian space Minkowski sum size estimation problem, and perhaps ultimately for convexity testing.

\subsection{Our Main Result}
Our main result is an algorithm for the subset size estimation problem which makes only \emph{constantly} many queries, independent of the ambient dimension $n$. We state our main result informally below:

\begin{inftheorem}
Given oracle access to any set $A \subseteq \mathbb{F}_2^n$ and an error parameter $\epsilon>0$, there is an algorithm making $O_{\epsilon}(1)$ queries to $A$ with the following guarantee:  with high probability, the algorithm outputs a value $0 \leq v \leq 1$ such that $\Vol(A' + A') - \eps \leq v \leq \Vol(A'+A')+\eps$ for some set $A' \subseteq A$ such that $\Vol(A \setminus A') \le \epsilon$.
\end{inftheorem} 

In fact, as we describe in more detail later, our algorithm does more than just approximate the volume of $A' + A'$:  it outputs a high-accuracy approximate oracle for the set $A' +A'$, given which it is trivially easy to approximate $\Vol(A' + A')$ by random sampling.  (As we will see, our algorithm also outputs an exact oracle for the set $A'$.)  Later we will give a formal definition of what it means to ``output an oracle'' for a set $B$; informally, it means we give a description of an oracle algorithm (which uses a black-box oracle to $A$) which, on any input $x$, (i) determines whether $x \in B$, and (ii) makes few invocations to the oracle for $A$.  We further note that the running time of our algorithm is {linear in $n$}   
 (note that even writing down an $n$-bit string as a query input to $A$ takes linear time).

\subsection{Technical Overview}

\subsubsection{A Conceptual Overview of the Algorithm}

In this subsection we give a technical overview of our algorithm. At a high level, our approach is based on the {\em structure versus randomness} paradigm that has proven to be very influential  in additive combinatorics~\cite{TaoVu:06}  and property testing. Our algorithm relies on two main ingredients, which we describe below.

To explain the key ingredients we need the notion of {quasirandomness} from additive combinatorics. For a set $A \subseteq \mathbb{F}_2^n$, we say $A$ is $\epsilon$-\emph{quasirandom} if each non-empty Fourier coefficient $\widehat{A}(\alpha), 0^n \neq \alpha \in \F_2^n$, satisfies $|\widehat{A}(\alpha)| \leq \eps,$ where we are viewing $A$ as a characteristic function over the domain $\F_2^n.$  The definition of the Fourier transform extends to the more general setting in which $A$ is a characteristic function whose domain is some coset $x + H$ (of size $2^{n-k}$) of $\F_2^n$. This is done by identifying $H$ with $\F_2^k$ via a homomorphism; we give details later  in~\Cref{def:quasirandom}.

The first ingredient is the following: Let $H$ be a linear subspace of $\mathbb{F}_2^n$, and let $B_x  \subseteq x+H$, $B_y \subseteq y+H$ be subsets of cosets $x+H$ and $y+H$ respectively. Suppose that both $|B_x|/|x+H|$ and $|B_y|/|y+H|$ are at least $\tau$, and that both $B_x$ and $B_y$ are $\epsilon$-quasirandom (viewed as characteristic functions whose domains are the cosets $x+H$ and $y+H$ respectively).  Our first ingredient is the simple but useful observation that if $\tau \gg \sqrt{\epsilon}$, then the set $B_x + B_y$ (which is easily seen to be a subset of the coset $x+y+H$) must be almost the entire coset $x+y + H$  (see \Cref{prop:anindya-bogolyubov}).

The second ingredient is Green's well-known ``regularity lemma'' for Boolean functions~\cite{Green:05}. To explain this, for any set $A \subseteq \F_2^n$, subspace $H$ of $\F_2^n$, and coset $H'$, let $A_{H'} := A \cap H'$ be the intersection of $A$ with the coset $H'$. Roughly speaking, Green's regularity lemma shows that for any $A \subseteq \F_2^n$, there is a  subspace $H$ of codimension at most $O_{\gamma,\eps}(1)$\ignore{\rnote{Maybe we save the special treat of the exact $2 \uparrow\uparrow \frac{1}{\gamma \epsilon^2}$ quantitative dependence for later? \shivam{Sounds good :)}}} such that the following holds: With probability $1-\gamma$ over a uniform random choice of cosets $\{H_i\}$, the set $A_{H_{i}}$ is $\epsilon$-quasirandom (viewed as a subset of the coset $H_i$). Moreover, the proof of the regularity lemma gives an iterative procedure to identify $H$; very roughly speaking, until the procedure terminates, at each stage it identifies a vector $\alpha \in \F_2^n$ such that $|\widehat{A}(\alpha)|$ is large, and sets $H$ to be the span of the vectors identified so far. 

With these two ingredients in place, we are ready to explain (at least at a qualitative level; we defer discussion of how to achieve the desired $O(1)$ query complexity to the next subsection) the algorithm for simulating an oracle to $A' + A'$.\ignore{which approximately computes the size of $A'+A'$.}\ignore{\rnote{This okay? Hopefully earlier we made it clear that really we're going to be simulating $A' + A'$ and that the $\Vol(A'+A')$ estimation is trivial given this}}  First, we run the algorithmic version of Green's regularity lemma; having done so, we have a subspace $H$ and we know that for most cosets $H'$, the set $A_{H'}$ is $\epsilon$-quasirandom. Let $k$ be the codimension of $H$ and let ${\cal B}'$ be a set of $2^k$ many coset representatives for the $2^k$ cosets of $H$. Let ${\cal B} \subseteq {\cal B}'$ be the subset consisting of those coset representatives $y \in {\cal B}'$ for which the set $A_{y+H}$ (i) is $\epsilon$-quasirandom and (ii) has density at least $\tau$ when viewed as a subset of $y+H$ (where $\tau$ is some carefully chosen parameter that we do not specify here). We note that given any coset $y+H$, condition (ii) can be checked using simple random sampling. Condition (i) is equivalent to checking that the set $A_{y+H}$ has no Fourier coefficient larger than $\epsilon$. This can be done using the celebrated Goldreich-Levin algorithm~\cite{goldreich-levin}.\footnote{To be more accurate, this requires a slight adaptation of the Goldreich-Levin algorithm because the domain here is a coset rather than the more familiar domain $\F_2^n$ for Goldreich-Levin.}  Thus, at this point our algorithm has determined the set ${\cal B} \subseteq {\cal B}'.$

The set $A' \subset A$ is defined to be
\[
A' := \bigcup_{y \in {\cal B}} A_{y+H},
\]
i.e.~$A'$ is obtained from $A$ by removing $A_{y+H}$ for each $y \in {\cal B}' \setminus {\cal B}$, or equivalently, ``zeroing out'' $A$ on every coset $y+H$ where $A_{y+H}$ either is not $\epsilon$-quasirandom or has density 
smaller than $\tau$. (Since the algorithm knows $H$ and ${\cal B}$, it is clear from this definition of $A'$ that, as mentioned after the informal theorem statement given earlier, the algorithm can simulate an exact oracle for the set $A'$.)  Turning to $A' + A'$, we have that
\begin{eqnarray}
A' + A' &=& \bigcup_{y,z \in \mathcal{B}} \pbra{A \cap (y+H)} + \pbra{A \cap (z+H)}, \nonumber \\
&\approx& \bigcup_{y,z \in \mathcal{B}} y+z+H, ~\label{eq:appx-oracle}
\end{eqnarray} 
where the last line follows from \Cref{prop:anindya-bogolyubov} (that we informally stated as the first ingredient mentioned above). As above, since the algorithm knows $H$ and ${\cal B}$, it is clear from that the algorithm can simulate an approximate oracle for $A' + A'$.

\subsubsection{Achieving Constant Query Complexity}

The above description essentially gives the high level description of our algorithm, at least at a conceptual level.  However, there is a significant caveat, which arises when we consider the query complexity of the algorithm.  Our goal is to achieve query complexity $O_{\epsilon}(1)$, but explicitly obtaining a description of the subspace $H$ necessarily requires a number of queries that scales at least linearly in $n$; indeed, even explicitly describing a single vector in $H$ requires $\Theta(n)$ bits of information (and thus this many queries).  Similarly, obtaining an explicit description of even a single vector $y \in {\cal B}'$ would be prohibitively expensive using only constantly many queries.  To circumvent these obstacles and achieve constant (rather than linear or worse) query complexity, we need to develop {``implicit''} versions of the procedures described above. 

As an example, we recall that the standard Goldreich-Levin algorithm, given oracle access to any set $A \subseteq \mathbb{F}_2^n$, outputs a list of parity functions $\chi_{\alpha^{(1)}}, \chi_{\alpha^{(2)}},\dots$ such that the Fourier coefficient $|\widehat{A}(\alpha^{(i)})|$ is ``large'' (roughly, at least $\eps$) for each $i$.  However, explicitly outputting the label $\alpha^{(i)}$ of even a single parity would require $n$ bits of information. To avoid this, we slightly modify the standard Goldreich-Levin procedure to show that with $\poly(1/\epsilon)$ queries, we can output {\em oracles} to the parity functions $\chi_{\alpha^{(1)}}, \chi_{\alpha^{(2)}},\dots$. In turn, each such oracle can be computed on any point $x \in \F_2^n$ with just $\poly(1/\epsilon)$ many queries to the set $A$; thus, we have {\em implicit} access to the parity functions $\{\chi_\alpha\}$ rather than explicit descriptions of the parities.  In the language of coding theory, this amounts to an analysis showing that the Goldreich-Levin algorithm can be used to achieve constant-query ``local list correction'' of the Hadamard code.  We view this as essentially folklore \cite{Sudan21}; it is implicit in a number of previous works \cite{sudtrevad01,KS13}, but the closest explicit statements we have been able to find in the literature essentially say that Goldreich-Levin is a constant-query local list decoder (rather than local list corrector) for the Hadamard code. 

With an ``implicit'' version of the Goldreich-Levin algorithm in hand, we show how to carefully use this implicit Goldreich-Levin to obtain an ``implicit'' algorithmic version of Green's regularity lemma.  This implicit version is sufficient to carry out the steps mentioned above with overall constant query complexity. We hope that the implicit (query-efficient) versions of these algorithms may be useful in other settings beyond the current work.

\subsection{Related Work}  

As noted earlier, our sumset size estimation problem has a similar flavor to the work of \cite{KothariNOW14,Neeman2013surface} on testing surface area, but the technical details are entirely different.

We note that for any invertible affine transformation $\Phi: \F_2^n \to \F_2^n$, we have that $\Vol(A+A) = \Vol (\Phi A + \Phi A)$ (but clearly this need not hold for noninvertible affine transformations). Starting with the influential paper of Kaufman and Sudan~\cite{KaufmanSudan:08}, a number of works have studied the testability of affine-invariant properties, see e.g.~\cite{bhattacharyya2015unified, bhattacharyya2013testing, hatami2013estimating, yoshida2014characterization, hatami2016general, bhattacharyya2013guest} These works consider properties that are invariant under {\em all} affine transformations (not just invertible ones), which makes them inapplicable to our setting. However, we note that there are thematic similarities between the approaches in those works and our approach (in particular, the use of the {``structure versus randomness''} paradigm).

%% file: sections/prelims.tex

\section{Preliminaries} \label{sec:prelims}

In this section, we set notation and briefly recall preliminaries from additive combinatorics and Fourier analysis of Boolean functions. Given arbitrary $A,B\sse\F_2^n$, we define
\[\Vol(A) := \frac{|A|}{2^n}\qquad\text{and}\qquad\Vol_{B}(A) := \frac{\abs{A\cap B}}{|B|}.\]
We will sometimes identify 
a set $A\sse\F_2^n$ with its indicator function $A: \F_2^n \to \zo$, defined as
\[A(x) = \begin{cases}
 1 & x\in A\\ 0 	& x\notin A
 \end{cases}
\] 
for $x\in\F_2^n$. When $A \subseteq x+H$ for some coset $x+H$, we similarly identify $A$ with its indicator function $A: x+H \to \{0,1\}$.
We write $e_i\in\F_2^n$ to denote the vector with a $1$ in the $i^\text{th}$ position and $0$ everywhere else. The function $2\uparrow\uparrow m$ denotes an exponential tower of $2$'s of height $m$ and the function $\log^\ast$ denotes its inverse. 

\subsection{Analysis of Boolean Functions} 
\label{subsec:fourier}

Our notation and terminology follow \cite{odonnell-book}. We will view the vector space of functions $f:\F_2^n\to\R$ as a real inner product space, with inner product $\abra{f,g} := \Ex_{\bx\sim\F_2^n}\sbra{f(\bx)g(\bx)}$. It is easy to see that the collection of \emph{parity functions} $\cbra{\chi_\alpha}_{\alpha\in\F_2^n}$ where $\chi_\alpha(x) := (-1)^{\abra{\alpha, x}} = (-1)^{\sum_{i=1}^n \alpha_i x_i}$ forms an orthonormal basis for this vector space. In particular, every function $f : \F_2^n \to \R$ can be uniquely expressed by its \emph{Fourier transform}, given by
\begin{equation} \label{eq:fourier-expansion}
	f(x) = \sum_{\alpha\in\F_2^n} \wh{f}(\alpha)\chi_\alpha(x).
\end{equation}
The real number $\wh{f}(\alpha)$ is called the \emph{Fourier coefficient of $f$ on $\alpha$}, and the collection of all $2^n$ Fourier coefficients of $f$ is called the \emph{Fourier spectrum} of $f$. 
We recall\ignore{It is easy to see that} Parseval's and Plancherel's formulas: for all $f, g : \F_2^n \to \R$, we have
\begin{equation} \label{eq:parseval-plancharel}
	\abra{f,f} = \sum_{\alpha\in\F_2^n} \wh{f}(\alpha)^2 \qquad\text{and}\qquad \abra{f,g} = \sum_{\alpha\in\F_2^n} \wh{f}(\alpha)\wh{g}(\alpha).
\end{equation}
It follows that $\E[f] = \wh{f}(0)$. Given $f, g : \F_2^n \to \R$, their \emph{convolution} is the function $f\ast g : \F_2^n \to \R$ defined by 
\[f\ast g(x) := \Ex_{\by\sim\F_2^n}\sbra{f(\by)g(x+\by)},\]
which satisfies
\begin{equation} \label{eq:convolution-formula}
	\wh{f\ast g}(\alpha) = \wh{f}(\alpha)\cdot\wh{g}(\alpha).
\end{equation}

\subsection{Subspaces and Functions on Subspaces}
\label{subsubsec:subspaces}

Throughout this subsection, let $f: \F_2^n\to\R$ and let $H\leq\F_2^n$ be a linear subspace of codimension $k$ (so $|H|=2^{n-k}$). We can write 
\begin{equation} \label{eq:H-def}
	H = \cbra{x : \abra{x,\alpha_i} = 0\ \forall\ i\in\{1,\ldots,k\}}
\end{equation}
for some linearly independent collection of vectors $\{\alpha_1,\dots,\alpha_k\}$. 

A coset\ignore{\rnote{The notation was $H_i$ for the coset but that clashed slightly with the dummy index in the next centered equation, and I didn't see a reason to use $H_i$ here so for now I am making this $H'$ instead 
\shivam{Makes sense, thanks!}}} $H'$, which is an affine subspace or, equivalently, a ``translate'' $y+H$ for some $y\in\F_2^n$, can be expressed as a set of the form 
\begin{equation} \nonumber \label{eq:H-coset}
	H' = \cbra{x : \abra{x,\alpha_i} = b_i\ \forall\ i\in\{1,\ldots,k\}}
\end{equation}
for some $b_i \in \F_2$; we will often identify $H'$ with the vector $b := (b_1, \ldots, b_k)$. Note that if $H' = y+H$, then $b_i = \abra{y,\alpha_i}$. 

Any coset of $H$ is affinely isomorphic to a copy of $\F_2^{n-k}$, and this lets us define the Fourier transform of a function $f:\F_2^n\to\R$ restricted to a coset $H'$.  More formally, consider the function $f_{H'} : H'\ignore{\cong \F_2^{n-k}} \to \R$ defined as $f_{H'}(x) = f(x)$. Its Fourier spectrum is indexed by the {$2^{n-k}$ elements of $H$;}\ignore{collection of $2^{n-k}$ vectors $\{\beta \in \F_2^n: \beta\perp \mathrm{span}\cbra{\alpha_1,\ldots,\alpha_k}\}$.} in particular, for each $\beta \in H$ we have 
\begin{equation} \label{eq:coset-fourier-coeff}
	\wh{f_{H'}}(\beta) = \frac{1}{2^{n-k}}\sum_{x\in H'} f(x)\chi_{\beta}(x).
\end{equation}

We can alternatively restrict a function $f:\F_2^n\to\R$ to a coset $H'$, but treat it as a function on $\F_2^n$ that takes value 0 on all points in $\F_2^n \setminus H'$; this viewpoint will be notationally cleaner to work with going forward so we elaborate on it here. We define the function $f\restriction_{H'}: \F_2^n\to\R$ as 
	\begin{equation} \label{eq:restriction-fourier-2}
		f\restriction_{H'}(x) = 
		\begin{cases}
 			f(x) & x\in H'\\
 			0 & \text{otherwise}
		\end{cases}.
	\end{equation}
The Fourier coefficients of $f_{H'}$ and $f\restriction_{H'}$ are related by the following simple fact. 

\begin{fact} \label{fact:annoying-coset-stuff}
	Let {{$f:\F_2^n \rightarrow \mathbb{R}$}}, $H$ be as in \Cref{eq:H-def}, and let $H'$ be a coset of $H$.		{Let ${\cal B}' \subseteq \F_2^n, |{\cal B}'|=2^k$ be a collection of $2^k$ coset representatives for $H$ (so every vector in $\F_2^n$ has a unique representation as $\gamma + \beta$ for some $\gamma \in {\cal B}', \beta \in H$).}
	For any $ \gamma\in{\cal B}', \beta \in H$, we have 
	\[\abs{\wh{f\restriction_{H'}}(\gamma+\beta)} = \frac{1}{2^k}\cdot\abs{\wh{f_{H'}}(\beta)}.\]
\end{fact}

\begin{proof}
For ease of notation we first consider the case that $\alpha_i = e_i$. Suppose that 
	$H'$ is given by 
	\[H' = \cbra{x : \abra{x,e_i} = b_i\ \forall\ i\in\{1,\ldots,k\}}\] 
	where we write $b := (b_1,\ldots, b_k)$.  We may take ${\cal B}'$ to be the set of all $2^k$ vectors in $\F_2^n$ whose last $n-k$ coordinates are all 0, and we note that $H=\mathrm{span}\cbra{e_{k+1},\dots,e_n}.$
	
	For $ \gamma\in{\cal B}', \beta \in H$, we have 
	\begin{align*}
		\wh{f\restriction_{H'}}(\gamma + \beta) &= \frac{1}{2^n}\sum_{x\in\F_2^n}f\restriction_{H'}(x)\chi_{\gamma + \beta}(x)\\
		&= \frac{1}{2^n}\sum_{x_1\in\F_2^k}\sum_{x_2\in\F_2^{n-k}}f\restriction_{H'}(x_1, x_2)\chi_\gamma(x_1) \chi_\beta(x_2)\\
		\intertext{where we have abused notation in the last line and viewed $\gamma \in \F_2^k, \beta \in \F_2^{n-k}$.  In turn the above is equal to}
		&= \frac{1}{2^n}\sum_{x_2\in\F_2^{n-k}}f\restriction_{H'}(b, x_2)\chi_{\gamma}(b)\chi_{\beta}(x_2)\\
\intertext{as $(x_1,x_2)\notin H'$ (and hence $f\restriction_{H'}(x_1, x_2)=0$) if $x_1\neq b$, and so}
		&= \frac{\chi_{\gamma}(b)}{2^n} \sum_{x_2\in\F_2^{n-k}} f\restriction_{H'}(b, x_2)\chi_{\beta}(x_2)\\
		&= \frac{\chi_{\gamma}(b)}{2^k}\cdot\frac{1}{2^{n-k}}\sum_{x_2\in\F_2^{n-k}} f_{H'}({b},x_2)\chi_{\beta}(x_2) \\
\intertext{\ignore{where we abuse notation and view $\beta \in \F_2^{n-k}$, which lets us write}which by \Cref{eq:coset-fourier-coeff} gives us} 
		&= \frac{\chi_{\gamma}(b)}{2^k}\cdot\wh{f_{H'}}(\beta).
	\end{align*}
	The result in the general case follows by applying an invertible linear transformation mapping $\alpha_i \mapsto e_i$ (see Exercise~3.1 of \cite{odonnell-book}). 
\end{proof}

\subsection{Parity Decision Trees}

We will only need the notion of a ``nonadaptive'' parity decision tree:

\begin{definition}[nonadaptive parity decision tree] \label{def:parity-decision-tree}
A \emph{nonadaptive parity decision tree} $\calT_f$ is a representation of a function $f:\F_2^n\to\R$. It consists of a rooted binary tree of depth $d$ with $2^d$ leaves, so every root-to-leaf path has length exactly $d$. Each internal node at depth $i$ is is labeled by a vector $\alpha_i \in\F_2^n$ corresponding to the parity function $\chi_{\alpha_i}(\cdot)$, and the vectors $\alpha_1,\dots,\alpha_d \in \F_2^n$ are linearly independent. (Having all nodes at level $i$ be labeled with the same vector $\alpha_i$ is the sense in which the tree is ``nonadaptive.'')
The outgoing edges of each internal node are labeled $0$ and $1$, and the leaves of $\calT_f$ are labeled by functions (which are restrictions of $f$). The \emph{size} of $\calT_f$ is the number of leaf nodes of $\calT_f$. 

In more detail, a root-to-leaf path can be written as $\{(\alpha_i\to b_i)\}$ where we follow the outgoing edge $b_i$ from the internal node $\alpha_i$, with $b_i\in\F_2$. 
On an input $x$, the parity decision tree $\calT_f$ follows the root-to-leaf path $\{(\alpha_i\to \abra{\alpha_i, x})\}$ and outputs the value of the function associated to the leaf at $x$. 
\end{definition}

Note that given $f: \F_2^n\to\R$ and a subspace $H\leq\F_2^n$ of codimension $k$ as in \Cref{eq:H-def}, we can associate a natural parity decision tree $\calT_f$ in which each level-$i$ internal node is labeled by $\alpha_i$ and each leaf node (corresponding to some coset $H'$ of $H$) is labeled by $f\restriction_{H'}$.

\subsection{Quasirandomness and Green's Regularity Lemma}
\label{subsec:additive-combo}


The following definition of \emph{quasirandomnesss} has been well-studied as a notion of pseudorandomness in additive combinatorics; we refer the interested reader to \cite{Chung1992} for more details.  

\begin{definition}[$\eps$-quasirandomness] \label{def:quasirandom}
	We say that $f:\F_2^n\to\R$ is $\eps$-\emph{quasirandom} if 
	\[\sup_{0^n\neq\alpha}\abs{\wh{f}(\alpha)} \leq \eps.\]
\end{definition}

\begin{definition}[$\eps$-quasirandom when restricted to coset] \label{def:quasirandom-coset}
	Let $f:\F_2^n\to\R$, $H\leq\F_2^n$ as in \Cref{eq:H-def}, and let $H'$ be a coset of $H$. We say that $f_{H'}:H'\to\R$ is $\eps$-\emph{quasirandom} if
	\[\sup_{0^n\neq\beta {\in H}\ignore{\perp\mathrm{span}\{\alpha_i\}}}\abs{\wh{f_{H'}}(\beta)} \leq \eps\]
	where $\wh{f_{H'}}(\beta)$ is as defined in \Cref{eq:H-coset}.
\end{definition}

In \Cref{def:quasirandom,def:quasirandom-coset}, the function of interest will often be the indicator of a subset $A\sse\F_2^n$. We next state Green's regularity lemma for Boolean functions, which is analogous to Szemer\'edi's celebrated graph regularity lemma \cite{Szemeredi:78}. 

\begin{proposition}[Green's regularity lemma in $\F_2^n$] \label{prop:green-reg-lemma}
	Let $A\sse\F_2^n$ and $\eps > 0$. There exists a subspace $H\leq\F_2^n$ with cosets $\{H_i\}$ such that
	\begin{enumerate}
		\item the codimension of $H$ is at most $2\uparrow\uparrow \frac{1}{\gamma\eps^2}$; and 
		\item for all but $\gamma$-fraction of cosets of $H$, the function $A_{H_i} : H_i \to \zo$ is $\eps$-quasirandom.
	\end{enumerate}
\end{proposition} 

In \Cref{sec:implicit-stuff}, we will see the proof of Green's regularity lemma (in the course of providing a constructive and highly query-efficient version of the lemma).\ignore{ which is more, the proof of which will closely follow Green's original proof of {\red{\Cref{prop:green-reg-lemma}}} itself \cite{Green:05}.}

\subsection{The Goldreich--Levin Theorem}
\label{subsec:prelims-gl}

Given \emph{query access} to a function $f: \F_2^n\to \zo$, the Goldreich--Levin algorithm \cite{goldreich-levin} allows us to find all linear (parity) functions that are well-correlated with $f$ (equivalently, it allows us to find all the ``significant'' Fourier coefficients of $f$). More formally, we have the following result. 

\begin{proposition}[Goldreich--Levin algorithm] \label{thm:km}
	Let $A\sse\F_2^n$ be arbitrary and let $\theta, \delta > 0$ be fixed. There is an algorithm $\GL\pbra{A, \theta, \delta}$ that, given query access to $A : \F_2^n\to\zo$, outputs a subset $\calS \sse \F_2^n$ of size $O\pbra{1/\theta^2}$ such that with probability at least $1 - \delta$, we have
	\begin{itemize}
		\item if $\alpha \in \calS$, then $\abs{\wh{A}(\alpha)} \geq \frac{\theta}{2}$; and 
		\item if $\abs{\wh{A}(\alpha)} \geq \theta$, then $\alpha \in \calS$. 
	\end{itemize}
	Furthermore, $\GL\pbra{A, \theta, \delta}$ runs in $\poly\pbra{n, \frac{1}{\theta}, \log\frac{1}{\delta}}$ time and makes $\poly\pbra{n, \frac{1}{\theta}, \log\frac{1}{\delta}}$ queries to $A$.
\end{proposition}

\subsection{Oracles and Oracle Machines}
\label{subsec:oracles}

As stated in the introduction, the outputs of our algorithmic procedures---\Cref{alg:constructive-regularity-lemma,alg:simulate-sumset}---will be \emph{oracles} to the indicator functions of specific subsets of $\F_2^n$. We first recall the definition of a probabilistic oracle machine:

\begin{definition} \label{def:prob-oracle-machine}
Let $f: \F_2^n \to \zo$.  
A randomized algorithm ${\cal O}$ with black-box query access to $f$, denoted ${\cal O}^f$,  is said to be a \emph{probabilistic oracle machine for $g: \F_2^n \to \zo$} if for any input $x \in \F_2^n$, the algorithm ${\cal O}^f$ outputs a bit ${\cal O}^f(x)$ that satisfies
 \[
 \Pr [\mathcal{O}^f(x)  = g(x)] \ge 2/3,
 \]
where the probability is taken over the internal coin tosses of $\mathcal{O}^f$. The \emph{query complexity} of the machine is the number of oracle calls made by $\mathcal{O}$ to $f$ and the \emph{running time} of the machine is the number of time steps it takes in the worst case (counting each oracle call as a single time step).
\end{definition}

Of course, the 2/3 in the above definition can be upgraded to $1-\tau$ at a cost of increasing the query complexity by a factor of $O(\log (1/\tau))$. We next define  what it means for an algorithm to ``output an (approximate) oracle'' for a function.

\begin{definition}
	Let $f,g$ be two functions $f,g: \F_2^n\to \{0,1\}$. An algorithm $\calA$ with query access to $f$, denoted by $\calA^f$, is said to \emph{output a $\pbra{\delta, q, T}$-oracle $\calO^f_g$ for the function $g$} if it outputs a representation of a {probabilistic} oracle machine $\calO^f_g: \F_2^n\to\R$ for which the following hold:
	\begin{enumerate}
		\item We have $\dist(\calO^f_g, g) \leq \delta$ (i.e. $\Pr_{\bx \sim \F_2^n}[\calO^f_g(\bx) \neq g(\bx)] \leq \delta$);
		\item The query complexity of ${\cal O}^f_g$ is at most $q$ and the running time of ${\cal O}^f_g$ is at most $T$.
	\end{enumerate}
	If $\delta = 0$, then we say that $\calO^f_g$ is an \emph{exact oracle} for $g$. 
\end{definition} 

%% file: sections/not-so-glorious-algo.tex

\section{A Query-Inefficient Version of the Main Result}
\label{sec:algo-inefficient}

\newcommand{\Treg}{\calT_{\mathrm{regular}}}

In this section, we prove a query-inefficient ``non-implicit'' version of our main result, which has a polynomial query complexity dependence on the ambient dimension $n$. 
In particular, we will prove the following theorem. 

\begin{theorem} [Main result, query-inefficient version] \label{thm:main-non-implicit}
	Let $A\sse\F_2^n$ be an arbitrary subset, and let $\eps, \tau > 0$. Given query access to $A$, there exists an algorithm that makes $\poly\pbra{n, 2\uparrow\uparrow\frac{8}{\eps^3}, \frac{1}{\tau}}$ queries to $A$ and does a $\poly\pbra{n, 2\uparrow\uparrow\frac{8}{\eps^3}, \frac{1}{\tau}}$ time computation and outputs with probability at least $9/10$:
	\begin{enumerate}
		\item A $\pbra{0, 1, O(n)}$-oracle $\calO^A_{A'}$ to the indicator function of $A'\sse A$ where $\Vol\pbra{A\setminus A'} \leq \eps + \tau$; and
		\item A $\pbra{O\pbra{\eps^2/\tau^4},0, O(n)}$-oracle $\calO^A_{A'+A'}$ to the indicator function of the sumset $A'+A'$.
	\end{enumerate}
\end{theorem}

In \Cref{sec:implicit-stuff}, we will present an ``implicit'' version of \Cref{thm:main-non-implicit} that makes only $O_\eps(1)$ queries, independent of the ambient dimension $n$, and thereby prove our main result.

We start by recording a corollary of Green's regularity lemma in \Cref{subsec:green-decomp}, which (informally), given an arbitrary set $A\sse\F_2^n$, establishes the existence of a ``structured'' set $A'\sse A$ capturing ``almost all'' of $A$. \Cref{subsec:algorithmic-green-decomp} then presents a procedure---\ConstructDT---that constructs an \emph{exact} oracle to this structured set $A'$, giving item (1) of the above theorem. In \Cref{subsec:approx-sumset-simulation}, we present a procedure---\Sumset---that constructs an \emph{approximate} oracle to the sumset $A'+A'$, giving item (2).

\subsection{Partitioning Arbitrary Sets into Dense Quasirandom Cosets}
\label{subsec:green-decomp}

Green's regularity lemma in $\F_2^n$ says that given an arbitrary set $A\sse\F_2^n$ and an error parameter $\eps > 0$, we can partition $\F_2^n$ into $O_\eps(1)$ (independent of $n$) many sets such that $A$ is ``random-like'' on almost all of these sets. Moreover, all these sets have a convenient structure: they are cosets of a common subspace of constant codimension.

We will use the following easy consequence of Green's lemma:

\begin{proposition} \label{prop:green-corollary}
	Given $A\sse\F_2^n$ and $\eps, \tau > 0$, there exists a subspace $H\leq\F_2^n$ of codimension at most $2\uparrow\uparrow\frac{1}{\eps^3}$ and a set $A'\sse A$ such that 
	\begin{enumerate}
		\item $\Vol\pbra{A\setminus A'} \leq \eps + \tau$;
		\item For any coset $H_i$, either $\Vol_{H_i}(A') = 0$ or $\Vol_{H_i}(A') \geq \tau$; and
		\item $A'_{H_i}$ is $\eps$-quasirandom for all cosets $H_i$.
	\end{enumerate}
\end{proposition} 

\begin{proof}
	Let $H\leq \F_2^n$ be the subspace of of codimension at most $2 \uparrow\uparrow \frac{1}{\eps^3}$ guaranteed to exist by \Cref{prop:green-reg-lemma}, and let $\cbra{H_1, \ldots, H_M}$ be an enumeration of the cosets of $H$ where $M = 2^n\cdot\abs{H}^{-1}$. We know from \Cref{prop:green-reg-lemma} that for all but $\eps$-fraction of $\cbra{H_i}$, the function $A_{H_i} : H_i \to \zo$ is $\eps$-quasirandom. 
	
	Define disjoint subsets $A'_1,\dots,A'_M$, where each $A'_i \subseteq A \cap H_i$, as follows:
\begin{enumerate}
	\item If $A_{H_i}$ is not $\eps$-quasirandom, then $A'_i = \emptyset$; 
	\item If $\Vol_{H_i}(A) \leq \tau$, set $A'_i = \emptyset$;
	\item Otherwise, set $A'_i = A \cap H_i$. 
\end{enumerate}
We now define $A' \sse \F_2^n$ as
\begin{equation} \label{eq:A-prime-decomposition}
	A' := \bigsqcup_{i=1}^M A'_i.
\end{equation}
We clearly have $\Vol(A\setminus A') \leq \eps + \tau$ and that $A_{H_i}$ is $\eps$-quasirandom for all $i \in [M]$. (Note that ${\emptyset}$ is trivially $\eps$-quasirandom.)
\end{proof}

Informally, \Cref{prop:green-corollary} modifies $A$ to obtain a structured set $A'\sse A$ that contains ``most'' of $A$ and has either empty or ``large'' intersection with all of the cosets guaranteed to exists by Green's regularity lemma. Furthermore, $A'$ is ``random-like'' on \emph{all}---as opposed to \emph{almost all}---of these cosets. 

\subsection{A Constructive Regularity Lemma via the Goldreich--Levin Theorem}
\label{subsec:algorithmic-green-decomp}

In this section, we make \Cref{prop:construct-dt-correctness} constructive via the Goldreich--Levin algorithm. The procedure \ConstructDT\ presented in \Cref{alg:constructive-regularity-lemma} closely follows the structure of Green's original proof of the regularity lemma itself \cite{Green:05}.

\begin{algorithm}[!htp] 
\caption{A constructive regularity lemma in $\F_2^n$} 
\label{alg:constructive-regularity-lemma}

\textbf{Input:} Query access to $A\sse\F_2^n$, quasirandomness parameter $\eps$, density threshold $\tau$
\vspace{0.5em}

\noindent\textbf{Output:} An exact oracle $\calO_{A'}^A$ and a parity decision tree $\Treg$ with $A'$ as in \Cref{prop:green-corollary}

\

\ConstructDT$\pbra{A,\eps,\tau}$:

\begin{enumerate}[rightmargin=1cm]
	\item 
	Initialize the decision tree $\Treg$ to contain no internal nodes and one leaf labelled by $A: \F_2^n \to \zo$. Define 
	\[\delta := \pbra{2\uparrow\uparrow\frac{8}{\eps^3}}^{-1}\cdot\frac{1}{30}\]
	\item At each stage of growing $\Treg$,  do the following: 
	\begin{enumerate}
		\item Let $\cbra{H_1, \ldots, H_M}$ denote the cosets corresponding to the leaves of the decision tree at the current stage. The $i^\text{th}$ leaf node is labelled by the function $A_{H_i}$.  
		\item For each coset $H_i$,  call 
			\[\calS_{i} \leftarrow \GL\pbra{A\restriction_{H_i} , {\eps/M}, \delta}.\]
		\item For each non-empty $\calS_i$, for each $\alpha\in\calS_i$, estimate $\abs{\wh{A_{H_i}}(\alpha)}$ up to additive error $\pm\eps/4$ with confidence $\delta$. If the estimate is less than $3\eps/4$, then remove $\alpha$ from $\calS_i$. 
		\item If $\calS_{i} = \emptyset$ for at least $(1-\eps)$-fraction of the $\cbra{\calS_1, \ldots, \calS_M}$, go to Step 3.
		\item Let the collection of labels of all internal nodes be $\calL$. For each non-empty $S_i$:
			\begin{enumerate}
				\item Choose $\alpha \leftarrow \calS_i$. Check if the collection $\calL \cup \{\alpha\}$ is linearly independent. 
				\item If so, then add $\alpha$ to $\calL$ and split all nodes at the current stage on $\alpha$.$^\dagger$
			\end{enumerate}
		\item Repeat Step 2. 
	\end{enumerate}
	\item For each leaf node---say, corresponding to the coset $H_i$---estimate $\wh{\Theta}_{i} := \Vol_{H_i}(A)$ up to an additive error of $\pm\tau/4$ with confidence $\delta$.
	\begin{enumerate}
		\item If $\wh{\Theta}_i \geq 3\tau/4$, set the function associated to the leaf node to be the identically-$1$ function.
		\item Else set it to be the identically-$0$ function. 
	\end{enumerate}
	\item Define the oracle $\calO_{A'}^A$ to be the function 
	\[\calO_{A'}^A(x) = \Treg(x)\cdot A(x).\]
\end{enumerate}

\

\hrule 

\

\quad {\footnotesize $\dagger$ By ``splitting'' a leaf node on a parity $\alpha\in\F_2^n$, we mean replacing it with an internal node labeled by the parity $\alpha$ with two natural leaf nodes as children.}
\end{algorithm}

\begin{proposition} \label{prop:construct-dt-correctness}
	Let $A\sse\F_2^n$ be an arbitrary subset. Given query access to $A$ and $\eps, \tau > 0$, the procedure \ConstructDT$(A,\eps,\tau)$ described in \Cref{alg:constructive-regularity-lemma}:
	\begin{enumerate}
		\item Makes $\poly\pbra{n, 2\uparrow\uparrow\frac{8}{\eps^3}, \frac{1}{\tau}}$ queries to $A$ and does a $\poly\pbra{n, 2\uparrow\uparrow\frac{8}{\eps^3}, \frac{1}{\tau}}$ time computation; and 
		\item With probability $9/10$ outputs a deterministic $\pbra{0, 1, O(n)}$-oracle $\calO_{A'}^A$ for $A'$ where $A'\sse A$ is as in \Cref{prop:green-corollary}.
			\end{enumerate}
\end{proposition}

We note that the procedure \ConstructDT\ makes queries to the oracle $A$ in the course of running the Goldreich--Levin algorithm.

\newcommand{\potfunc}{\mathrm{ExpImb}}

\begin{proof}
	
	We first argue that Step~2 in the procedure \ConstructDT\ terminates; this essentially follows from Green's original proof of the regularity lemma in $\F_2^n$. In particular, suppose, at the current stage, the subspace given by the internal nodes of the parity decision tree is $H$, and let $\{H_1, \ldots, H_M\}$ denote the cosets corresponding to the leaves. Consider the potential function
	\[\potfunc[A, H] := \frac{1}{M}\sum_{i=1}^M |\wh{A_{H_i}}\pbra{0^n}|^2,\]
	where we recall that $\wh{A_{H_i}}\pbra{0^n} = \Vol_{H_i}(A)$. Note that $\potfunc[A, H]\in [0,1]$. Informally, $\potfunc$ captures the ``expected imbalance'' of $A$ restricted to the leaf nodes of the tree at the current stage. 
	
	Lemma~2.2 of \cite{Green:05} (alternatively, see \cite{ROD-green-regularity}) states if there exists a leaf node $A_{H_i}$ and a parity $\alpha\in\F_2^n$ such that $\abs{\wh{A_{H_i}}(\alpha)} \geq \eps/2$, then upon splitting all nodes at the current level on the parity $\alpha$---with $H'\leq H$ being the subspace corresponding to the resulting tree---we have 
	\[\potfunc[A,H'] \geq \potfunc[A,H] + \frac{\eps^3}{4}.\]
	It follows that if the condition in Line~2(d) of \ConstructDT doesn't hold, then after Step~2(e), the value of $\potfunc$ increases by at least $4/\eps^3$. It follows that Step~2 can be repeated at most $2\uparrow\uparrow\frac{4}{\epsilon^3}$ times.
		
	Next, note that the Goldreich--Levin call in Step~2(b) makes at most $\poly\pbra{n, 2\uparrow\uparrow\frac{8}{\eps^3}, \frac{1}{\tau}}$ queries to $A$ over the run of \ConstructDT, and each call to Step~2(e) and Step~3 makes $O\pbra{\frac{1}{\eps^2}}$ and $O\pbra{\frac{1}{\tau^2}}$ many queries (via a standard application of the Chernoff bound). The overall query complexity of \ConstructDT\ follows. The runtime is similarly clear. 
		
	Note that we run the Goldreich--Levin algorithm in Step~2(b) on the function $A\restriction_{H_i}$ as opposed to $A_{H_i}$. It follows from \Cref{fact:annoying-coset-stuff} that $A\restriction_{H_i}$ is $\eps/M$-quasirandom if and only if $A_{H_i}$ is $\eps$-quasirandom (where $M$ is the number of cosets at a particular stage of the algorithm). We also note that given query access to $A_{H_i}$, we can simulate query access to $A\restriction_{H_i}$ by checking whether an input $x$ belongs to the coset $H_i$ by querying it on the parity decision tree $\Treg$. 
	
	In the pruning procedure in Step~2(e), the size of each $\calS_i$ is at most $O(1/\eps^2)$. A union bound over the Goldreich--Levin and estimation procedures implies that with probability $9/10$, the function computed by $\Treg$ indicates whether a point $x$ is in a coset $H'$ for which $A_{H'}$ is $\eps$-quasirandom and also $\Vol_{H'}(A) \geq \tau$. It follows that $\calO_{A'}^A$ is an exact oracle for $A'$; it also clearly makes exactly $1$ query to $A$ on any input. 
\end{proof}

\subsection{Approximately Simulating Sumsets}
\label{subsec:approx-sumset-simulation}

Note that \Cref{prop:green-corollary} asserts, for arbitrary $A\sse\F_2^n$, the existence of a structured subset $A'\sse A$ (which is ``almost all of $A$'') and a subspace $H \leq \F_2^n$ such that $A'^{+y}_H$ is $\eps$-quasirandom for all $y \in \F_2^n$. The following lemma indicates why such a decomposition is useful towards our goal of (approximately) simulating sumsets. 
 
\begin{lemma} \label{prop:anindya-bogolyubov}
	Let $A\sse\F_2^n$ be arbitrary and let $H\leq \F_2^n$ be a subspace. Suppose, for $x, y \in \F_2^n$, 
	\begin{enumerate}
		\item $A_{x+H},A{}_{y+H}$ are $\eps$-quasirandom (in the sense of \Cref{def:quasirandom-coset}); and 
		\item $\Vol_{x+H}\pbra{A}, \Vol_{y+H}\pbra{A} \geq \tau$ for some $\tau > 0$.
	\end{enumerate}
	Then we have
	\begin{equation} \label{eq:anindya-bogolyubov}
		\Vol_{x+y+H}\pbra{A + A}\geq 1 - O\pbra{\frac{\eps^2}{\tau^4}}.
	\end{equation}
\end{lemma}

\begin{proof}
	For ease of notation, define the $\eps$-quasirandom functions $f : x+H \to \zo$ and $g: y+H \to \zo$ as
	\[f:=A_{x+H} \qquad\text{and}\qquad g:=A_{y+H}.\]
	Consider $h := f\ast g$ and note that $\supp(h) = A_{x+H} + A_{y+H}$. From \Cref{eq:convolution-formula}, we have that 
	\begin{equation} \label{eq:h-decomp}
		h = \sum_{\alpha\in H} \wh{f}(\alpha)\wh{g}(\alpha)\chi_\alpha \geq \tau^2 + \underbrace{\sum_{0^n\neq\alpha\in H} \wh{f}(\alpha)\wh{g}(\alpha)\chi_\alpha}_{=: \Gamma}.
	\end{equation}
	Note that $\Ex_{\bx\sim H}\sbra{\Gamma(\bx)} = 0$ and 
	\begin{align*}
		\Ex_{\bx\sim H}\sbra{\Gamma(\bx)^2} = \sum_{0^n\neq\alpha\in H}\wh{f}(\alpha)^2\wh{g}(\alpha)^2
		\leq \max_{0^n\neq\alpha\in H}\wh{f}(\alpha)^2  \pbra{\sum_{0^n\neq\alpha\in H}\wh{g}(\alpha)^2}
		\leq \eps^2
	\end{align*}
	as $f$ is $\eps$-quasirandom. It then follows from Chebyshev's inequality that  
	\[\Prx_{\bx\sim H}\sbra{\abs{\Gamma(\bx)} \geq \frac{\tau^2}{2}} = O\pbra{\frac{\eps^2}{\tau^4}}\qquad\text{and so}\qquad
	\Prx_{\bx\sim H}\sbra{h(\bx) > 0} \geq 1 - O\pbra{\frac{\eps^2}{\tau^4}},
	\]
	completing the proof. 
\end{proof}

\begin{remark}
Note that the lower bound of $1-O(\eps^2/\tau^4)$ in \Cref{eq:anindya-bogolyubov} cannot be improved to $1$, as witnessed by the following example: Let $A \sse \F_2^n$ be defined as
\[A(x) = \begin{cases}
 1 & \sum x_i \geq \frac{n}{2}\\
 0 & \sum x_i \leq \frac{n}{2} - 1	
 \end{cases}
\]
and let $H = \F_2^n$. As $A$ is a symmetric function, $\wh{A}(\alpha)$ only depends on $\sum_i \alpha_i$. It is easy to check using Parseval's identity that $\abs{\wh{A}(\alpha)}\leq O\pbra{\frac{1}{\sqrt{n}}}$, and that $A + A \subsetneq \F_2^n$ (as we clearly have $1^n \notin A + A$).
\end{remark}

\Cref{prop:anindya-bogolyubov} suggests a natural approach towards our goal of approximately simulating sumsets: Given the parity decision tree $\Treg$ as in \Cref{alg:constructive-regularity-lemma}, for every pair of leaves---say, corresponding to cosets $x+H$ and $y+H$---with non-trivial $\Vol_{x+H}(A'),\Vol_{y+H}(A')$, we set $\Vol_{x+y+H}\pbra{A'+A'} = 1$. This procedure is outlined in \Cref{alg:simulate-sumset}; more formally, we have \Cref{prop:approx-sumset-simul}.  

\begin{algorithm}[ht]
\caption{Approximately simulating query access to the sumset $A'+A'$} 
\label{alg:simulate-sumset}

\textbf{Input:} Query access to $A\sse\F_2^n$, quasirandomness parameter $\eps$, density threshold $\tau$
\vspace{0.5em}

\noindent\textbf{Output:} An approximate oracle $\calO_{A'+A'}$ where $A'$ as in \Cref{prop:green-corollary}

\

\Sumset$\pbra{A,\eps,\tau}$:

\begin{enumerate}[rightmargin=1cm]
	\item Obtain $\pbra{\calO_{A'}^A, \Treg}$ via 
	\[\pbra{\calO_{A'}^A, \Treg} \leftarrow \ConstructDT\pbra{A, \eps, \tau}.\] 
	Let $\alpha_i\in\F_2^n$ denote the label associated to internal nodes at depth $i$. 
	\item We will write $\pbra{b_1,\ldots,b_k}$ with $b_i\in\F_2$ to denote the root-to-leaf path obtained by taking the outgoing edge labeled by $b_i$ from the internal node $\alpha_i$,and will identify leaves of $\Treg$ with these tuples. 
	\item Initialize $\Tsum$ as a copy of $\Treg$, and associate all leaves with the identically-$0$ function. 
	\item For all pairs of leaf nodes $(b^{(1)}_1, \ldots, b^{(1)}_k)$ and $(b^{(2)}_1, \ldots, b^{(2)}_k)$ in $\Treg$: 
	\begin{enumerate}
		\item If for both of the leaf nodes in the pair, the function associated with the leaf node is not the identically-$0$, function, then set the function associated to the leaf node $(b^{(1)}_1 + b^{(2)}_1, \ldots, b^{(1)}_k + b^{(2)}_k)$ in $\Tsum$ to be the identically-$1$ function.
	\end{enumerate}
	\item Define the oracle $O_{A'+A'}$ to be the function computed by $\Tsum$. 
\end{enumerate}

\end{algorithm}

\begin{proposition} \label{prop:approx-sumset-simul}
	Let $A\sse\F_2^n$ be an arbitrary subset. Given query access to $A$, and $\eps, \tau > 0$, let $A'\sse A$ as in \Cref{prop:green-corollary}. The procedure \Sumset$(A,\eps,\tau)$ described in \Cref{alg:simulate-sumset}:
	\begin{enumerate}
		\item Makes $\poly\pbra{n, 2\uparrow\uparrow\frac{8}{\eps^3}, \frac{1}{\tau}}$ queries to $A$ and does a $\poly\pbra{n, 2\uparrow\uparrow\frac{8}{\eps^3}, \frac{1}{\tau}}$ time computation; and 
		\item With probability $9/10$, outputs an $\pbra{O\pbra{\eps^2/\tau^4},0, n}$-oracle $\calO^A_{A'+A'}$ for $A'+A'$.
	\end{enumerate}
\end{proposition}

\begin{proof}
	Note that the number of queries made to $A$ follows from \Cref{prop:construct-dt-correctness}, and the runtime is immediate from Step 4. The second item above follows from \Cref{prop:anindya-bogolyubov}. 
\end{proof}

Note that \Cref{thm:main-non-implicit} follows immediately from \Cref{prop:construct-dt-correctness,prop:approx-sumset-simul}. Furthermore, we can easily estimate $\Vol\pbra{A'+A'}$ (where $A'$ as in \Cref{thm:main-non-implicit}) via random sampling. A standard application of the Chernoff bound shows that $O\pbra{\log(1/\delta)/\gamma^2}$ many samples suffice to get a $\pm\gamma$ additive approximation to $\Vol\pbra{A'+A'}$ with probability at least $1-\delta$.

%% file: sections/tbil-km-rocco-version.tex

\section{An Implicit Regularity Lemma in $\F_2^n$}
\label{sec:implicit-stuff}

In this section, we present the following ``implicit'' version of \Cref{thm:main-non-implicit} that makes that makes only $O_\eps(1)$ queries, independent of the ambient dimension $n$. 

\begin{theorem} [Main theorem] \label{thm:main-thm-implicit}
	Let $A\sse\F_2^n$ be an arbitrary subset, and let $\eps, \tau > 0$. Given query access to $A$, there exists an algorithm that makes $O_{\eps,\tau}(1)$ queries to $A$ and does an $O_{\eps,\tau}(1)\cdot n$ time computation and outputs with probability at least $9/10$:
	\begin{enumerate}
		\item A $\pbra{0, O_{\eps,\tau}(1), O_{\eps,\tau}(1) \cdot n)}$-oracle $\calO^A_{A'}$ to the indicator function of $A'\sse A$ where $\Vol\pbra{A\setminus A'} \leq \eps + \tau$; and
		\item A $\pbra{O\pbra{\eps^2/\tau^4},O_{\eps,\tau}(1), O_{\eps,\tau}(1) \cdot n)}$-oracle $\calO^A_{A'+A'}$ to the indicator function of the sumset $A'+A'$.
	\end{enumerate}
\end{theorem}

In \Cref{subsec:implicit-gl}, we state an ``implicit'' version of the Goldreich--Levin algorithm (which appears to be a folklore result in coding theory), which we then use in \Cref{subsec:eff-green,subsec:eff-sumset} to obtain query-efficient versions of \Cref{alg:constructive-regularity-lemma,alg:simulate-sumset} in \Cref{subsec:eff-green,subsec:eff-sumset} respectively. 

\subsection{Implicitly Finding Significant Fourier Coefficients}
\label{subsec:implicit-gl}

To explain what we mean by the qualifier ``implicit'', recall  the usual statement of the Goldreich-Levin algorithm (\Cref{thm:km}). Informally, the theorem states that given oracle access to $A \subseteq  \mathbb{F}_2^n$, there exists an algorithm that outputs an \emph{explicit} set $\mathcal{S} \subseteq \mathbb{F}_2^n$ with the elements of $S$ corresponding to the ``significant'' Fourier coefficients of $A$. In the language of coding theory, this a {\em list decoding algorithm} for the Hadamard code. 
 
The theorem as stated, however, is not useful for us---in particular, as our target query complexity is independent of $n$, we cannot hope to obtain $\mathcal{S}$ explicitly. We will instead obtain {\em implicit access} to the set $\mathcal{S}$. We next state the refined guarantee for the Goldreich-Levin algorithm that we require.

\ignore{
\gray{Towards this, we define the notion of a {\em probabilistic oracle machine}. 
 
 \begin{definition}
Corresponding to any $\alpha \in \mathbb{F}_2^n$, let $\chi_\alpha: \mathbb{F}_2^n \rightarrow \mathbb{F}_2$ be the corresponding parity function.  
We call $\mathcal{O}^A$ a probabilistic oracle machine for $\chi_\alpha$, if for any $x \in \mathbb{F}_2^n$, 
 \[
 \Pr [\mathcal{O}^A(x)  = \chi_\alpha(x)] \ge 2/3,
 \]
where the probability is taken over the internal coin tosses of $\mathcal{O}^A$. The query complexity of the machine is the number of oracle calls made by $\mathcal{O}$ to $A$. Note that the probability in the right hand side can be made $1-\delta$ at a cost of increasing the query complexity by a factor of $O(\log (1/\delta))$. 
 \end{definition} 
 }
 } 
 
 \begin{theorem}[Implicit Goldreich--Levin theorem] \label{thm:gl-refined}
 Given oracle access to set $A \subseteq \mathbb{F}_2^n$, significance threshold $\theta$ and confidence parameter $\delta$, the algorithm \ImplicitGL$(A,\eps,\tau,\delta)$ makes $\poly(1/\theta) \cdot \log (1/\delta)$ queries to $A$ and with probability at least $1-\delta$ for some $T \le 4/\theta^2$, outputs $T$ oracle machines $\mathcal{O}_1^A, \ldots, \mathcal{O}_T^A$ with the following guarantee: 
 \begin{enumerate}
 \item For each $1\le i \le T$, there is a distinct $\alpha_i \in \mathbb{F}_2^n$ such that $\mathcal{O}_i^A$ is a probabilistic oracle machine for the function $\chi_{\alpha_i}$. The query complexity of each oracle $\mathcal{O}_i^A$ is $\poly(1/\theta)$.
 \item For each $1 \le i \le T$, $|\widehat{A}(\alpha_i)| \ge \theta/2$. 
 \item For any $\beta \in \mathbb{F}_2^n$ such that $|\widehat{A}(\beta)| \ge \theta$, there is a $1 \le j \le T$, such that $\alpha_j =\beta$. 
 \end{enumerate}

 \end{theorem}
The crucial feature of \Cref{thm:gl-refined} is that the query complexity of both the routine \ImplicitGL~ as well as the probabilistic oracle machines is independent of $n$ and is just dependent on the significance parameter $\theta$ and confidence parameter $\delta$. Further, note that 
the algorithm \ImplicitGL~ does not just give an oracle for $\alpha_i$ (which would be a procedure which, on input $j \in [n]$, outputs the value of the $j$-th coordinate of $\alpha_i \in \F_2$), but rather it gives an oracle for $\chi_{\alpha_i}$ (which of course, on input $x \in \F_2^n$, outputs the value of $\chi_{\alpha_i}(x) \in \F_2$). In the parlance of coding theory,  \ImplicitGL~ is a {\em constant query} algorithm for {\em local list correction}. 

We note that in the usual formulation of Goldreich-Levin~(see~\cite{goldreich-levin,arora2009computational}), the algorithm outputs all the parities, i.e., the entire set $\mathcal{S}$. As the description size of $\mathcal{S}$ is $\Omega(n)$, the query complexity is necessarily $\Omega(n)$.
However, the formulation in \Cref{thm:gl-refined} can easily be obtained by the obvious modification of Rackoff's analysis~\cite{gol01} of the Goldreich-Levin algorithm
and seems to be folklore in coding theory~\cite{Sudan21}. In fact, a weaker statement, namely that Goldreich-Levin  is a {\em constant query local list decoding} algorithm has already been explicitly noted in literature~\cite{trevisan2004some, KS13}. 

We describe the routine \ImplicitGL~ in detail in \Cref{alg:implicit-GL}. To do so, we first need to define the procedure \BLR. 
\begin{definition}
The procedure $\BLR(\mathcal{D}, \tau_c, \tau_\ell, \kappa)$ takes as input oracle access to $\mathcal{D}: \mathbb{F}_2^n \rightarrow \mathbb{F}_2$, distance parameters $\tau_c < \tau_\ell $ and confidence parameter $\kappa$. With probability $1-\kappa$, $\BLR$ can distinguish between the cases (i) $\mathcal{D}$ is $\tau_c$-close to some parity $\chi$ and (ii) $\mathcal{D}$ is $\tau_\ell$-far from every parity $\chi$. 
\end{definition}
The Fourier analysis based proof~\cite{BCHKS} of the standard linearity tester~\cite{BLR93} can be (easily) used to obtain such a procedure \BLR~ as long as $\tau_c <\tau_\ell/3$. The query complexity of the procedure is 
$\log (1/\kappa) \cdot \mathrm{poly}(1/|\tau_\ell-3\tau_c|)$. 


{At a high level, the routine \ImplicitGL~ starts exactly the same way as of the Goldreich-Levin algorithm---in particular, the standard analysis of Goldreich-Levin shows the following (for Step~3(b)): For any $\alpha$ such that $|\widehat{f}(\alpha)| >\theta$, there is some $b \in \mathbb{F}_2^t$ such that 
\[
\Prx_{\bx \sim \mathbb{F}_2^n} \sbra{\chi_\alpha(\bx) \not = \mathcal{D}_b^A(\bx)} \le  1/10. 
\] 
It easily follows that for any $\alpha$ such that $|\widehat{f}(\alpha)| >\theta$, there is some $b \in \mathbb{F}_2^t$ such that $\calO_b^A$ is a probabilistic oracle for $\chi_\alpha$. Further, for any $b \in \mathbb{F}_2^t$, $\calO_b^A$ is a probabilistic oracle for some parity. In Step~3(d), we compute the correlation between $\calO_b^A$ and $A$ up to $\pm \theta/4$. This implies that all $\calO_b^A$  which survive  satisfy properties (2) and (3) of \Cref{thm:gl-refined}. We leave the detailed analysis to the interested reader. 
 }






\begin{algorithm}[!htp]
\caption{An Implicit Goldreich--Levin Algorithm} 
\label{alg:implicit-GL}

\textbf{Input:} Query access to $A:\F_2^n\to\zo$, confidence parameter $\delta>0$ and significance parameter $\theta>0$. 
\vspace{0.5em}

\noindent\textbf{Output:} Probabilistic oracles $\mathcal{O}^A_1, \ldots , \mathcal{O}^A_T$ for some $T \le 4/\theta^2$. The oracle machines $\mathcal{O}^A_1, \ldots , \mathcal{O}^A_T$ satisfy conditions (1) and (2) from \Cref{thm:gl-refined}. 

\

{\ImplicitGL$\pbra{A,\theta,\delta}$}:

\begin{enumerate}[rightmargin=1cm]
	\item Let 
	\[t = 
	\log \bigg( \frac{1}{\theta^2} \bigg) + O(1)\] 
	and initialize $\calS = \emptyset$. 
	\item Pick $X_1,\ldots,X_t \in \F_2^n$ uniformly at random.
	\item For all $b := (b_1,\ldots, b_t) \in \F_2^t$:
	\begin{enumerate}
		\item For all $\emptyset\neq S\sse[t]$:
		\begin{enumerate}
			\item Define $X^S := \sum_{i\in S} X_i$.
			\item Define $b^S := \sum_{i\in S} b_i$.
		\end{enumerate}
		\item Define $\calD_b^A: \F_2^n\to\zo$ as 
		\[\calD_b^A(x) := \maj_{\emptyset\neq S\sse[t]} \bigg\{ A\pbra{X^S+x} + b^S \bigg\}. \]

		\item Define ${{\delta_1:= \frac{\delta \theta^2}{4}}}$. Run $\BLR\pbra{\calD_b^A, 1/20, 1/5, \delta_1}$. If $\BLR$ does not accept, discard $\calD_b^A$.

\item Define $\calO_b^A: \mathbb{F}_2^n \rightarrow \{0,1\}$ as follows: choose $y_1, \ldots, y_{R} \in \mathbb{F}_2^n$ where $R = \Theta(\log(1/\delta))$. 
\[
\calO_b^A(x)  := \maj_{1 \le j \le R} \bigg\{ \calD_b^A(x+y) + \calD_b^A(y) \bigg\}.
\]
	\item Estimate $\wh{\Theta}_b := \abra{A, \calO_b^A}$ up to an additive error of $\pm\theta/4$ and confidence $\delta_1$. If the estimate $\wh{\Theta}_b < 3\theta/4$, discard $\mathcal{O}_b^A$.		
	\end{enumerate}
	\item Output all $\mathcal{O}_b^A$ which survive. 
\end{enumerate}

\end{algorithm}

\subsection{A Query-Efficient Version of \Cref{alg:constructive-regularity-lemma}}
\label{subsec:eff-green}

Recall that \Cref{alg:constructive-regularity-lemma} takes in as input query access to $A\sse\F_2^n$, a quasirandomness parameter $\eps >0 $, and a density threshold $\tau > 0$, and outputs a $\pbra{0, 1, O(n)}$-oracle $\calO_{A'}^A$ for $A'$ where $A'\sse A$ is as in \Cref{prop:green-corollary}. The value of this oracle $\calO_{A'}^A$ on an input $x \in \F_2^n$ is obtained by routing $x$ through the decision tree $\Treg(x)$ (recall that each internal node of $\Treg(x)$ is labeled by an ``explicit'' parity function $\chi_{\alpha_i}$ obtained from some call to the  \GL~algorithm), and outputting the value $A(x)$ if a 1-leaf is reached (if a $0$-leaf is reached the output is 0). This call to $A$ at the leaf that $x$ reaches is why $\calO_{A'}^A$ makes one (and only one) call to the oracle for $A$.

In contrast, in the query-efficient regime we cannot use the standard  \GL~algorithm because its $\Omega(n)$ query complexity is prohibitively high; instead we replace each call to \GL~with a call to \ImplicitGL.  While \GL~returns explicit parity functions which label the various nodes of $\Treg$, the \ConstructImplicitDT~procedure constructs an ``implicit'' decision tree in which each node queries some probabilistic oracle machine (that was returned by \ImplicitGL) to obtain the value of the desired parity function.  Consequently, a call to the oracle $\calO_{A'}^A$ produced by \ConstructImplicitDT~makes $d \cdot \ell + 1$ calls to $A$, where $d$ is the depth of the implicit decision tree and $\ell$ is the number of oracle calls to $A$ that are made by each parity oracle produced by \ImplicitGL. Crucially, both $d$ and $\ell$ are values that are $O_{\eps,\tau}(1)$ and completely independent of $n$.

 In addition to constructing an implicit decision tree, \ConstructImplicitDT~also needs to check for linear independence of the obtained parity oracles (see Step~2(e)(i) of \Cref{alg:constructive-regularity-lemma}) in a query-efficient way.  We detail the performance guarantee of \ConstructImplicitDT~in the following proposition:
 
\begin{proposition} \label{prop:implicit-green-lemma}
	Let $A\sse\F_2^n$ be an arbitrary subset. Given query access to $A$ and $\eps, \tau > 0$, there exists an algorithm \ConstructImplicitDT~that:
	\begin{enumerate}
		\item Makes $O_{\eps,\tau}(1)$ queries to $A$ and does an $O_{\eps,\tau}(1)$ time computation; and 
		\item With probability $9/10$, outputs a probabilistic $\pbra{0, O_{\eps,\tau}(1), O(n)}$-oracle $\calO_{A'}^A$ for $A'$ where $A'\sse A$ is as in \Cref{prop:green-corollary}.	\end{enumerate}
\end{proposition}

\begin{proof}
The \ConstructImplicitDT~procedure is obtained by modifying the \ConstructDT~procedure presented in \Cref{alg:constructive-regularity-lemma} in the following ways.

	\begin{enumerate}
	
	\item In Line~2(a) of \ConstructDT, instead of maintaining a list of explicit cosets $H_1,\dots,H_M$, the algorithm maintains a list of probabilistic oracle machines ${\cal O}^A_1,\dots,{\cal O}^A_{\log M}$ (obtained from calls to \ImplicitGL) for the $\log M$ parities which define the cosets $H_1,\dots,H_M$.
	
	\item In Line~2(b), to simulate access to $A \restriction_{H_i}$ on an input $x$, the algorithm queries the $\log M$ oracle machines and uses the obtained responses to determine whether or not $x$ belongs to the relevant coset.\ignore{we write $A \restriction_{\vec{{\cal O}}^A_i}$ to denote this simulated version of $A \restriction_{H_i}$.} In addition, each call to \GL$\pbra{A\restriction_{H_i}, \epsilon/M, \delta}$ in Step~2(b) is replaced with a call to \ImplicitGL$\pbra{A\restriction_{H_i} , \epsilon/M, \delta}$.  Note that each set ${\cal S}_i$ produced by a call to \ImplicitGL~is now a set of oracles for parity functions.

	\item Each estimate of $\abs{\wh{A_{H_i}}(\alpha)}$ in Line~2(c) is obtained by random sampling, using the simulated version of $A \restriction_{H_i}$ described above and the oracle for the parity function for $\chi_\alpha.$

		\item In Line~2(e)(i), since the algorithm does not explicitly have the vectors in $\F_2^n$ that define the parity functions, it instead use the following simple sampling-based procedure to check linear independence:
		
		\begin{itemize}
			\item Given a collection of oracles $\{\calO_{\chi_1}, \calO_{\chi_2},\ldots,\calO_{\chi_k}\}$ where each $\chi_i$ is some parity function, the algorithm queries all of them on $N$ independent uniform random points in $\F_2^n$ and builds the corresponding $k \times N$ matrix with entries in $\F_2.$
			\item Then the algorithm checks if the rank of this matrix is $k$.
		\end{itemize}
		It is clear that if the parities $\{\alpha_i\}$ are not linearly independent, then the $k\times N$ matrix constructed by this procedure will not have rank $k$. On the other hand, a simple probabilistic argument shows that if the $k$ parities are linearly independent, then the matrix will have rank $k$ except with failure probability at most $2^{k^2 - N}$.
	
	\item In Line~3 the estimate of $\Vol_{H_i}(A)$ is obtained using the $\log M$ oracle machines mentioned above in the obvious way; and
	
	\item Finally, the output oracle ${\cal O}^A_{A'}$ is the obvious analogue of $\Treg \cdot A$ where again the $\log M$ oracle machines are used to route inputs through the implicit decision tree to the correct coset.
		\end{enumerate}

	The analysis of correctness is essentially the same as that of \Cref{prop:construct-dt-correctness}. We note that while \ConstructDT~outputs a deterministic oracle, \ConstructImplicitDT~outputs a probabilistic oracle (because of the probabilistic oracles for parity functions that it uses).  For the query complexity, a tedious but straightforward inductive argument shows that the values of $\delta$ (for each call to \ImplicitGL) and $N$ (for each execution of Line~2(e)(i)) can be taken to be independent of $n$, yielding the claimed query complexity.
\end{proof}

\subsection{A Query-Efficient Version of \Cref{alg:simulate-sumset}}
\label{subsec:eff-sumset}

Finally, the query-efficient version of \Cref{alg:simulate-sumset}, which we call \ImplicitSumset, works in the obvious way. In Line~1, the call to \ConstructDT~is replaced by a call to \ConstructImplicitDT, and the ``explicit'' decision tree $\Treg$ is replaced by the ensemble of parity oracles corresponding to the coset decomposition.  We observe that while in the explicit algorithm \Sumset, the function $\Tsum$ can be evaluated on an input $x \in \F_2^n$ without making any calls to $A$, in our implicit setting we need to query the ensemble of parity oracles (and hence make queries to $A$) for each evaluation of $\Tsum$ on an input $x$ (to route $x$ to the correct leaf node in the implicit tree for $\Tsum$).
\Cref{thm:main-thm-implicit} follows from \Cref{prop:implicit-green-lemma}
 and the obvious analogue of \Cref{prop:approx-sumset-simul} for \ImplicitSumset.

%% file: sections/conclusion.tex

\def\Conv{\mathrm{Conv}}

\section{Conclusion and Future Work}

Our results suggest a number of interesting directions for future work.   In particular, a broad goal is to develop  query-efficient procedures for simulating oracle access to other types of sumsets, or sumsets over other domains.  Our approach extends relatively straightforwardly to the sumset $A+B$ for distinct sets $A,B \subseteq \F_2^n$ given access to oracles to both $A$ and $B$, and likewise to the iterated sumset $A + \cdots + A = kA$ for any constant $k$.  

A more ambitious extension would be to handle the sumset $A+A$ when $A$ is an arbitrary subset of some other Abelian (or potentially non-Abelian) group $G$.  Green's regularity lemma is known to hold for general finite Abelian groups \cite{Green:05}, but to obtain constant query complexity independent of $|G|$ via our approach it seems that one would need an ``implicit'' procedure for finding large Fourier coefficients of functions from $G$ to $\R$. As observed in \cite{DGKS08}, the algorithm of Goldreich and Levin does not generalize to finding large Fourier coefficients over arbitrary finite groups. There is an alternative algorithm, due to Kushilevitz and Mansour~\cite{kusman93}, for finding large Fourier coefficients of functions $\F_2^n \to \R$ which has been generalized to arbitrary finite Abelian groups $G$ by Akavia et al.~\cite{akagolsaf03}, but the query complexity of the Kushilevitz-Mansour algorithm grows with $n$ and the query complexity of the Akavia et al.~algorithm grows with $|G|.$ Developing a constant-query ``implicit'' version of the algorithm of Akavia et al. for general finite groups is an interesting specific direction for future work.  

Yet another intriguing problem, as mentioned in \Cref{sec:motivation}, is to try to develop a query-efficient algorithm for simulating an oracle to ${\frac {A+A} 2}$ {(where addition denotes the Minkowski sum)} or $\Conv(A)$ (the convex hull of $A$) when $A$ is a subset of $\R^n$ (and we view $\R^n$ as endowed with the standard Normal $\calN(0,1)^n$ distribution).